\documentclass[letterpaper, 10 pt, conference]{ieeeconf}

\usepackage{graphicx}
\usepackage{epstopdf}
\usepackage{epsfig}
\usepackage{times}
\usepackage{amsmath}
\usepackage{amssymb}
\usepackage{latexsym}
\usepackage{color}
\usepackage{verbatim}
\usepackage{cite}
\usepackage{BernsteinStyle2}
\usepackage{caption}
\usepackage[hidelinks]{hyperref}
\usepackage{makecell}

\makeatletter
\newcommand*\bigcdot{\mathpalette\bigcdot@{.5}}
\newcommand*\bigcdot@[2]{\mathbin{\vcenter{\hbox{\scalebox{#2}{$\m@th#1\bullet$}}}}}
\makeatother

\IEEEoverridecommandlockouts

\title{\LARGE
Adaptive Real-Time Numerical Differentiation\\
with Variable-Rate Forgetting and Exponential Resetting
}

\author{Shashank Verma, Brian Lai, and Dennis S. Bernstein%
\thanks{$^{*}$Shashank Verma, Brian Lai, and Dennis S. Bernstein are with the Department of Aerospace Engineering, University of Michigan, Ann Arbor, MI 48109, USA 
{\tt\small shaaero@umich.edu}}%
}

\begin{document}

\maketitle
\thispagestyle{empty}
\pagestyle{empty}

\begin{abstract}
Digital PID control requires a differencing operation to implement the D gain.
In order to suppress the effects of noisy data, the traditional approach is to filter the data, where the frequency response of the filter is adjusted manually based on the characteristics of the sensor noise.
The present paper considers the case where the characteristics of the sensor noise change over time in an unknown way.
This problem is addressed by applying adaptive real-time numerical differentiation based on adaptive input and state estimation (AISE).
The contribution of this paper is to extend AISE to include variable-rate forgetting with exponential resetting, which allows AISE to more rapidly respond to changing noise characteristics while enforcing the boundedness of the covariance matrix used in recursive least squares.
\end{abstract}

\section{INTRODUCTION}

Because of its simplicity of tuning and implementation as well as its ability to follow setpoint commands and reject step disturbances, PID control is without doubt the most widely used feedback control algorithm \cite{PIDastrom,ODwyer2009handbook}.
Although PID controllers can be implemented by analog circuits, the overwhelming trend is to implement these controllers digitally using sampled data \cite{knospe2006pid,keel2003new,porter1992genetic,bennett2001past,borase2021review,PID2012,PIDastrom}.
Since sampled data cannot be differentiated, implementation of the D gain is typically realized using backward differencing \cite{aastrom2013computer}.
Unfortunately, backward differencing is susceptible to sensor noise, and the practical remedy is to filter the data.
%
%
The frequency response of the filter can be adjusted manually based on the characteristics of the sensor noise, and this approach is routinely used in practice.

The motivation for the present paper is the case where the characteristics of the sensor noise change over time in an unknown way.
In an autonomous vehicle, for example, the accuracy of camera images may change depending on atmospheric conditions such as lighting and fog.
When the sensor noise is nonstationary, the  accuracy of a manually adjusted filter may be initial adequate but can be degraded under changing conditions.
The present paper addresses this problem by applying adaptive real-time numerical differentiation to PID control.

Numerical differentiation is a longstanding problem in signal processing
\cite{cullum,savitzky1964smoothing,ramm2,Jauberteau2009,Stickel2010,zhao2013,knowles_methods,verma2022acc_rcie_hgo_sg,haimovich2022}.
For the present purposes, however, we require real-time numerical differentiation, where only present and past data are used to estimate the signal derivative.
In other words, causal numerical differentiation is needed, which rules out the use of noncausal differentiation methods, which use future data.

For real-time numerical differentiation, we consider the technique developed in \cite{verma2023realtime}.
This technique is based on adaptive input estimation augmented by adaptive state estimation.
The combined adaptive input and state estimation (AISE) method uses recursive least squares (RLS) to adjust the coefficients of the approximate inverse model for input estimation along with the noise covariances of the Kalman filter used for state estimation.

The present paper extends AISE for real-time numerical differentiation by replacing classical RLS optimization by RLS with variable-rate forgetting (VRF).
Several techniques have been proposed for RLS/VRF
\cite{islam2019recursive,adamRLS2020,ankitRLS2020,bruceNandSRLS,mohseni2022recursive,lai2022exponential}.  
Forgetting is motivated by the practical need to rapidly update the parameters of AISE, while VRF is used to mitigate the tendency of RLS to diverge in the absence of persistency \cite{dasgupta1987,ankitRLS2020}.

For real-time numerical differentiation with VRF, the present paper uses a combination of the F-test developed in \cite{mohseni2022recursive} and exponential resetting developed in \cite{lai2022exponential}, where, to improve numerical conditioning, the latter technique enforces an upper bound on the covariance matrix.

The contents of the paper are as follows.
Section \ref{sec:aie_ase} provides a description of the adaptive input and state estimation (AISE) algorithm, incorporating variable-rate forgetting with exponential resetting (VRF/ER) recursive least squares.
Section \ref{sec:num_example} investigates the efficacy of this technique through two applications. 
First, in subsection \ref{sec:PID_prob_formu}, AISE and AISE/VRF-ER are employed for digital PID sampled-data control of first-order-lag-plus-dead-time dynamics. 
Following this, in subsection \ref{subsec: carsim example}, AISE and AISE/VRF-ER are applied to relative position data obtained from a vehicle simulation to estimate the relative velocity by numerical differentiation.


\section{ADAPTIVE INPUT AND STATE ESTIMATION}  \label{sec:aie_ase}
We now apply adaptive input and state estimation (AISE) \cite{verma2023realtime} to real-time numerical differentiation.
Consider the linear discrete-time SISO system
\begin{align}
	x_{k+1} &=  A x_{k} + Bd_{k}, 	\label{state_eqn}\\
	y_k  &= C x_k + D_{2,k} v_k, \label{output_eqn}
\end{align}
where
$k\ge0$ is the step,
$x_k \in \mathbb R^{n}$ is the unknown state,
$d_k \in \mathbb R$ is unknown input,
$y_k \in \mathbb R$ is a measured output,
$v_k \in \mathbb R$ is standard white noise, 
and $D_{2,k}v_k \in \mathbb R$ is the sensor noise at time $t = kT_\rms$, where $T_\rms$ is the sample time.
The matrices $A \in \mathbb R^{n \times n}$, $B \in \mathbb R^{n \times 1}$, and $C \in \mathbb R^{1 \times n}$, are assumed to be known and $D_{2,k}$ is assumed to be unknown.
The sensor-noise covariance is define as $V_{2,k} \isdef D_{2,k} D_{2,k}^\rmT$.
The {goal} of adaptive input estimation (AIE) is to estimate $d_k$ and $x_k$.

In the application of AIE to real-time numerical differentiation, we use \eqref{state_eqn} and \eqref{output_eqn} to model a discrete-time integrator. As a result, AIE furnishes an estimate denoted by $\hat{d}_k$ for the derivative of the sampled output $y_k$. 
For single discrete-time differentiation,  the values are $A = 1, B = T_\rms,$ and $C=1$. 
%

\subsection{Input Estimation}
AIE comprises three subsystems: the Kalman filter forecast subsystem, the input-estimation subsystem, and the Kalman filter data-assimilation subsystem.
First, consider the Kalman filter forecast step
\begin{gather}
	x_{{\rm fc},k+1} = A x_{{\rm da},k} + B \hat{d}_{k},	\label{kalman_fc_state}\\
	y_{{\rm fc},k} =  C x_{{\rm fc},k}, \label{kalman_fc_output}\\
	z_k = y_{{\rm fc},k} - y_k, 		\label{innov_error}
\end{gather}
where
$x_{\rm da,k} \in \mathbb R^{n}$ is the data-assimilation state, 
$x_{{\rm fc},k} \in \mathbb R^{n}$ is the forecast state,
$\hat d_k$ is the estimate of $d_k$, 
$y_{\rmf\rmc,k} \in \mathbb R$ is the forecast output,
$z_k \in \mathbb R$ is the residual, and $x_{{\rm fc},0} = 0$.

Next, in order to obtain $\hat{d}_k$, the input-estimation subsystem of order $n_\rme$ is given by the exactly proper dynamics
\begin{align}
\hat{d}_k = \sum\limits_{i=1}^{n_\rme} P_{i,k} \hat{d}_{k-i} + \sum\limits_{i=0}^{n_\rme} Q_{i,k} z_{k-i}, \label{estimate_law1}
\end{align}
%
where $P_{i,k} \in \BBR$ and $Q_{i,k} \in \BBR$.
%
%
AIE minimize $z_{k}$ by updating $P_{i,k}$ and $Q_{i,k}$ as shown below.
The subsystem \eqref{estimate_law1} can be reformulated as
\begin{align}
\hat{d}_k=\Phi_k \theta_k, \label{estimate_law12}
\end{align}
where the estimated coefficient vector $\theta_k \in \mathbb{R}^{l_{\theta}}$ is defined by
\begin{align}
\hspace{-0.2cm}\theta_k \isdef \begin{bmatrix}
P_{1,k} & \cdots & P_{n_{\rme},k} & Q_{0,k} & \cdots & Q_{n_{\rme},k}
\end{bmatrix}^{\rmT},
\end{align}
the regressor matrix $\Phi_k \in \mathbb{R}^{1 \times l_{\theta}}$ is defined by
\begin{align}
	\hspace{-0.2cm}\Phi_k \isdef
		\begin{bmatrix}
			\hat{d}_{k-1} &
			\cdots &
			\hat{d}_{k-n_{\rme}} &
			z_k &
			\cdots &
			z_{k-n_{\rme}}
		\end{bmatrix},
\end{align}
and $l_\theta \isdef  2n_{\rme} +1$.
%
The subsystem \eqref{estimate_law1} can be written using backward shift operator $\bfq^{-1}$ as
\begin{align}
   \hat{d}_{k} = G_{\hat{d}z,k}(\bfq^{-1})z_k,
\end{align}
where
\begin{align}
    G_{\hat{d}z,k} &\isdef D_{\hat{d}z, k}^{-1}  \it{N}_{\hat{d}z,k}, \label{d_hat_z_tf} \\
    D_{\hat{d}z,k}(\bfq^{-1}) &\isdef I_{l_d}-P_{1,k}\bfq^{-1} - \cdots-P_{n_\rme,k}\bfq^{-n_\rme}, \label{d_hat_z_tf_D} \\
    N_{\hat{d}z, k}(\bfq^{-1}) &\isdef Q_{0,k} + Q_{1,k} \bfq^{-1}+\cdots+Q_{n_\rme,k}\bfq^{-n_\rme}. \label{d_hat_z_tf_N}
\end{align}
%
%
%

Next, define the filtered signals
\begin{align}
\Phi_{{\rm f},k} &\isdef G_{{\rm f}, k}(\bfq^{-1}) \Phi_{k}, \quad
\hat{d}_{{\rm f},k} \isdef G_{{\rm f}, k}(\bfq^{-1}) \hat{d}_{k}, \label{eq:filtdhat}
\end{align}
%
%
where, for all $k\ge 0$,
\begin{align}
G_{{\rm f}, k}(\bfq^{-1}) = \sum\limits_{i=1}^{n_{\rm f}} \bfq^{-i}H_{i,k}, \label{Gf}
\end{align}
\begin{align}
H_{i,k} &\isdef \left\{
\begin{array}{ll}
C B, & k\ge i=1,\\
C \overline{A}_{k-1}\cdots \overline{A}_{k-(i-1)}  B, & k\ge i \ge 2, \\
0, & i>k,
\end{array}
\right. 
\end{align}
and $\overline{A}_k \isdef A(I + K_{{\rm da},k}C)$, where $K_{{\rm da},k}$ is the Kalman filter gain given by \eqref{kalman_gain} below.
%
%
%
Furthermore, for all $k \ge 0$, define the {\it retrospective variable} $z_{{\rm r},k} \colon \BBR^{l_\theta} \rightarrow \BBR$ by
\begin{align}
z_{{\rm r},k}(\hat{\theta}) \isdef z_k -( \hat{d}_{{\rm f},k} - \Phi_{{\rm f},k}\hat{\theta} ), \label{eq:RetrPerfVar} 
\end{align}
and define the \textit{retrospective cost function} $\SJ_k \colon \BBR^{l_\theta} \rightarrow \BBR$ by
%
%
\begin{align*}
    \SJ_k(\hat{\theta}) \isdef (\hat{\theta} - \theta_0)^\rmT R_{\theta} (\hat{\theta} - \theta_0) + \sum\limits_{i=0}^k R_z z_{{\rm r},i}^{2}(\hat{\theta}) +  R_{\rmd} (\Phi_i\hat{\theta})^2
\end{align*}
where $R_{\theta}\in\BBR^{l_{\theta} \times l_{\theta}}$ is positive definite, $R_z\in(0,\infty)$, and $R_d\in(0,\infty)$.
Then, for all $k\ge 0$, the unique global minimizer $\theta_{k+1} \triangleq \argmin_{\hat{\theta} \in \BBR^{l_\theta}} \SJ_k(\hat{\theta})$ is given recursively by the RLS update equations \cite{islam2019recursive} as
\begin{align}
P_{k+1}^{-1} &= P_k^{-1} + \widetilde{\Phi}_k^\rmT \widetilde{R} \widetilde{\Phi}_k, \label{covariance_update} \\
\theta_{k+1} &= \theta_{k} - P_{k+1} \widetilde{\Phi}^{\rmT}_{k} \widetilde{R} (\widetilde{z}_{k} + \widetilde{\Phi}_{k} \theta_{k}), \label{theta_update}
\end{align}
where $P_0 \isdef R_\theta^{-1}$, for all $k \ge 0$, positive-definite $P_k \in \BBR^{l_\theta \times l_\theta}$ is the covariance matrix, and where, for all $k \ge 0$,
\begin{gather*}
\widetilde{\Phi}_k \isdef \begin{bmatrix}
   \Phi_{\rmf, k}  \\
   \Phi_k   \\
\end{bmatrix}, \quad 
\widetilde{z}_k \isdef \begin{bmatrix}
   z_k-\hat{d}_{{\rm f},k}  \\
   0   \\
\end{bmatrix}, \quad
\widetilde{R} \isdef \begin{bmatrix}
   R_z & 0  \\
   0 & R_{\rmd}   \\
\end{bmatrix}.
\end{gather*}
Hence, \eqref{covariance_update} and \eqref{theta_update} recursively update the input-estimation subsystem \eqref{estimate_law1}.
\subsection{Recursive Least Squares with Variable-Rate Forgetting and Exponential Resetting} \label{sec:rls_var}
While AISE based on the RLS update equations, given by \eqref{covariance_update} and \eqref{theta_update}, has shown promise in preliminary testing \cite{verma2023realtime}, a major drawback of RLS is the monotonicity of the eigenvalues of the covariance matrix, resulting in slowed adaptation after a large amount of data has been collected \cite{ortega2020modified,Salgado1988modified}.
To mitigate slow adaptation, we combine AISE with two recent advances in RLS, namely, variable-rate forgetting based on the \textit{F}-Test \cite{mohseni2022recursive}, and exponential resetting \cite{lai2022exponential}.
We call this combination \textit{variable-rate forgetting with exponential resetting (VRF-ER)} RLS. 

VRF-ER replaces the inverse covariance update equation \eqref{covariance_update} with, for all $k \ge 0$, 
\begin{align}
P_{k+1}^{-1} &= \lambda_k P_k^{-1} + (1-\lambda_k)R_\infty + \widetilde{\Phi}_k^\rmT \widetilde{R} \widetilde{\Phi}_k, \label{covariance_update errls}
\end{align}
where the positive-definite matrix $R_\infty \in \BBR^{l_\theta \times l_\theta}$ is the user-selected \textit{resetting matrix} and, for all $k \ge 0$, $\lambda_k \in (0,1]$ is the \textit{forgetting factor}.
A forgetting factor $\lambda_k < 1$ allows the eigenvalues of $P_k$ to decrease, resulting in continued adaptation of the input-estimation subsystem \eqref{estimate_law1}, even after a large amount of data has been collected \cite{aastrom1977theory}. 
On the other hand, the resetting matrix $R_\infty$ prevents the eigenvalues of $P_k$ from becoming too large when excitation is poor \cite{lai2022exponential}, a phenomenon known as covariance windup \cite{malik1991some}.
See \cite{lai2022exponential} for further details.

Next, variable-rate forgetting based on the \textit{F}-test \cite{mohseni2022recursive} is used, for all $k \ge 0$, to select the forgetting factor $\lambda_k \in (0,1]$.
For all $k \ge 0$, we define the \textit{residual error} at step $k$ as
\begin{align}
    \varepsilon_k \isdef \widetilde{z}_{k} + \widetilde{\Phi}_{k} \theta_{k} \in \BBR^2.
\end{align}
Note that the residual error is a metric of how well the input-estimation subsystem \eqref{estimate_law1} predicts the input one step into the future.
Furthermore, for all $k \ge 0$, define the sample mean of the residual errors over the past $\tau \ge 1$ steps as
\begin{align}
    \bar{\varepsilon}_{\tau,k} \isdef \frac{1}{\tau} \sum_{i=k - \tau + 1}^k {\varepsilon_i} \in \BBR^2,
\end{align}
and define the sample variance of the residual errors over the past $\tau$ steps as
\begin{align}
    \Sigma_{\tau,k} \isdef \frac{1}{\tau}  \sum_{i=k - \tau + 1}^k (\varepsilon_i - \bar{\varepsilon}_{\tau,k})(\varepsilon_i - \bar{\varepsilon}_{\tau,k})^\rmT \in \BBR^{2 \times 2}.
\end{align}

The approach in \cite{mohseni2022recursive} compares $\Sigma_{\tau_n,k}$ to $\Sigma_{\tau_d,k}$, where $\tau_n \ge 1$ is the short-term sample size and $\tau_d > \tau_n$ is the long-term sample size.
If the short-term variance $\Sigma_{\tau_n,k}$ is more statistically significant than the long term variance $\Sigma_{\tau_d,k}$ according to the Lawley-Hotelling trace approximation \cite{mckeon1974f}, then $\lambda_k < 1$ is chosen inversely proportional to the statistical significance.
If not, then $\lambda_k = 1$.
In particular, for all $k \ge 0$, the forgetting factor is selected as 
\begin{align}
    \lambda_k \isdef \frac{1}{1 + \eta g_k \mathbf{1}[g_k]}, \label{VRF_eq}
\end{align}
where $\eta \ge 0$ is a tuning parameter, $\mathbf{1} \colon \BBR \rightarrow \{0,1\}$ is the unit step function, and, for all $k \ge 0$,
\begin{align}
    g_k \isdef& \sqrt{  \frac{\tau_n}{\tau_d} \frac{\tr(\Sigma_{\tau_n,k} \Sigma_{\tau_d,k}^{-1})}{c} }
    - \sqrt{F_{2 \tau_n,b}^{-1} (1-\alpha)},
    \\
    a \isdef& \frac{(\tau_n + \tau_d - 3)(\tau_d -1)}{(\tau_d - 5)(\tau_d - 2)},
    \\
    b \isdef& 4 + \frac{2(\tau_n + 1)}{a-1}, \quad c \isdef \frac{2\tau_n(b-2)}{b(\tau_d - 3)},
\end{align}
and where $\alpha \in [0,1]$ is the significance level and $F_{2 \tau_n,b}^{-1} \colon [0,1] \rightarrow \BBR$ is the inverse cumulative distribution function of the \textit{F}-distribution with degrees of freedom $2 \tau_n$ and $b$. 
For further details, see \cite{mohseni2022recursive} and \cite{mckeon1974f}.

\subsection{State Estimation}

The forecast variable $x_{{\rm fc},k}$, given by \eqref{kalman_fc_state}, is used to obtain the estimate $x_{{\rm da},k}$ of $x_k$ given, for all $k \ge 0$, by the Kalman filter data-assimilation step
\begin{align}
x_{{\rm da},k} &= x_{{\rm fc},k} + K_{{\rm da},k} z_k, \label{kalman_da_state}
\end{align}
where the Kalman filter gain $K_{{\rm da},k} \in \mathbb R^{n}$, the data-assimilation error covariance $P_{{\rm da},k} \in \mathbb R^{n \times n},$
and the forecast error covariance $P_{\rmf,k+1} \in \mathbb R^{n \times n}$ are given by
\begin{align}
    K_{{\rm da},k} &= - P_{\rmf,k}C^{\rmT} ( C P_{\rmf,k} C^{\rmT} + V_{2,k}) ^{-1}, \label{kalman_gain} \\
    P_{{\rm da},k} &=  (I_{n}+K_{{\rm da},k}C) P_{\rmf,k},\label{Pda} \\
	P_{\rmf,k+1} &=  A P_{{\rm da},k}A^{\rmT} + V_{1,k}, \label{Pf}
\end{align}
where $V_{2,k} \in \mathbb R$ is the measurement covariance matrix and 
\begin{align}
    V_{1,k}\isdef & B\ {\rm var}\ (d_k-\hat{d}_k)B^\rmT \nonumber
    \\
    & + A\ {\rm cov}\ (x_k - x_{{\rm da},k},d_k-\hat{d}_k)B^\rmT \nonumber
    \\
    & + B\ {\rm cov}\ (d_k-\hat{d}_k,x_k - x_{{\rm da},k})A^\rmT
\end{align}
and $P_{\rmf,0} = 0.$ 
\subsection{Adaptive Input and State Estimation} \label{sec:AdapInptStateEst}


This section summarizes adaptive input and state estimation  (AISE). 
Assuming that, for all $k \ge 0$, $V_{1,k}$ and $V_{2,k}$ are unknown in (\ref{Pf}) and \eqref{kalman_gain},
%
%
%
the goal is to adapt ${V}_{{1,\rm adapt},k}$ and ${V}_{{2,\rm adapt},k}$ at each step $k$ to estimate $V_{1,k}$ and $V_{2,k}$, respectively.
To do this, we define, for all $k \ge 0$, the 
performance metric $J_k \colon \BBR^{n\times n} \times \BBR \rightarrow \BBR$ as
\begin{align}
  {J}_{k}({V}_{1},{V}_{2}) \isdef |\widehat{S}_{ k}-{S}_{ k}|, \label{J_daptmetric}
\end{align}
where $\widehat{S}_{ k}$ is the sample variance of $z_k$ over $[0,k]$ given by
\begin{align}
    \widehat{S}_{k} \isdef \cfrac{1}{k}\sum^{k}_{i=0}(z_i - \overline{z}_k)^2, \quad
    \overline{z}_k \isdef \cfrac{1}{k+1}\sum^{k}_{i=0}z_i,  \label{var_comp}
\end{align}
and ${S}_{k}$ is the variance of the residual $z_k$ given by the Kalman filter, defined as
\begin{align}
    {S}_{k} \isdef  C (A P_{{\rm da},k-1}A^{\rmT} + V_{1}) C^{\rm T} + V_{2}.  \label{var_inno}
\end{align}
%
%

%
%
%
%
%
%
%
%
For all $k \ge 0$, we assume that ${V}_{{1,\rm adapt},k}  \triangleq \eta_k I_n$ and we define $\eta_k \in \BBR$ and ${V}_{{2,\rm adapt},k}$ as
%
%
%
%
%
%
%
\begin{align}
       \eta_k,{V}_{{2,\rm adapt},k} &\isdef \underset{ \eta \in [\eta_{L},\eta_{\rmU}],{V}_{2} \ge 0}{\arg\min} \ J_k(\eta I_{n},V_{2}), \label{covmin}
\end{align} 

where
$0 \le \eta_{L} \le \eta_{\rmU}.$
Next, defining ${J}_{\rmf,k} \colon \BBR \rightarrow \BBR $ as
\begin{align}
    {J}_{\rmf,k}(V_{1}) \isdef \widehat{S}_{k} - C (A P_{{\rm da},k-1}A^{\rmT} + V_{1})  C^{\rm T} \label{J1_func}
\end{align}
and using \eqref{var_inno}, (\ref{J_daptmetric}) can be rewritten as
\begin{align}
    {J}_k({V}_{1},{V}_{2}) = |{J}_{\rmf,k}(V_{1})-V_{2}|. \label{J_daptmetric_V2}
\end{align}
We construct a set of positive values of ${J}_{\rmf,k}$ as 
%
%
\begin{align}
      \SJ_{\rmf,k} \isdef \{J_{\rmf,k}(\eta I_{n}) \colon J_{\rmf,k}(\eta I_{n}) > 0, \eta_{L} \le\eta \le\eta_{\rmU}\} \subseteq \BBR. \label{J_f_positive}
\end{align}
%
%
Finally, Proposition \ref{prop: eta_k and V2,adapt minimizer} gives a method to compute $\eta_k$ and ${V}_{{2,\rm adapt},k}$, defined in \eqref{covmin}.
%

\begin{prop}
\label{prop: eta_k and V2,adapt minimizer}
    Let $k \ge 0$ and let $\eta_k \in [\eta_L,\eta_U]$ and $V_{2,k} \ge 0$ be given by \eqref{covmin}.
    If $\SJ_{\rmf,k}$, defined in \eqref{J_f_positive}, is nonempty, then, for any $\beta \in [0,1]$, $\eta_k$ and $V_{2,k}$ are given by
    \begin{align}
        \eta_k &= \underset{\eta \in [\eta_L,\eta_U]}{\arg \min} \ |J_{\rmf,k}(\eta I_{n}) -  \widehat{J}_{\rmf,k}(\beta)|,\\
        {V}_{{2,\rm adapt},k} &= J_{\rmf,k}(\eta_k I_n),
        \label{v_2_opt_1}
    \end{align}
    where
\begin{align}
       \widehat{J}_{\rmf,k}(\beta) \isdef \beta \min \SJ_{\rmf,k}+(1-\beta)\max \SJ_{\rmf,k}, \label{alpha1}
\end{align}
    If $\SJ_{\rmf,k}$ is empty, then $\eta_k$ and $V_{2,k}$ are given by
    \begin{align}
        \eta_k &= \underset{\eta \in [\eta_L,\eta_U]}{\arg \min} \ |J_{\rmf,k}(\eta I_{n})|,\\
        {V}_{{2,\rm adapt},k} &= 0. \label{v_2_opt_1}
    \end{align}
\end{prop}
\begin{proof}
Under the constrain that $\eta_k$ and $V_{2,k}$ are non-negative, at step $k$, if $\SJ_{\rmf,k}$ in \eqref{J_f_positive} is non-empty, we compute $\widehat{J}_{\rmf,k}(\beta)$ for the user-provided $\beta$ in \eqref{alpha1}. Subsequently, we select $\eta_k \in [\eta_{L}, \eta_U]$ using a grid search to minimize
\begin{align*}
    \eta_k &= \underset{\eta \in [\eta_L,\eta_U]}{\arg \min} \ |J_{\rmf,k}(\eta I_{n}) - \widehat{J}_{\rmf,k}(\beta)|, \\
    {V}_{{1,\rm adapt},k} &= \eta_k I_n.
\end{align*}
We then choose ${V}_{{2,\rm adapt},k} = J_{\rmf,k}({V}_{{1,\rm adapt},k}) = J_{\rmf,k}( \eta_k I_n) > 0$ such that
\begin{align*}
    {J}_k({V}_{{1,\rm adapt},k},{V}_{{2,\rm adapt},k}) &= |{J}_{\rmf,k}({V}_{{1,\rm adapt},k})-{V}_{{2,\rm adapt},k}| \\
    & = |{J}_{\rmf,k}(\eta_k I_n) -   J_{\rmf,k}( \eta_k I_n) | \\
    & = 0. 
\end{align*}
This demonstrates that both $\eta_k$ and ${V}_{{2,\rm adapt},k} = J_{\rmf,k}( \eta_k I_n)$ constitute an optimal solution for \eqref{covmin}.

When $\SJ_{\rmf,k}$ is empty, indicating that $J_{\rmf,k}(\eta I_{n}) < 0$ for all $\eta \in [\eta_L, \eta_U]$. We minimize $J_k$ in \eqref{J_daptmetric_V2} by choosing ${V}_{{2,\rm adapt},k} = 0$. We then select $\eta_k \in [\eta_{L}, \eta_U]$ using a grid search such that
\begin{align*}
    \eta_k &= \underset{\eta \in [\eta_L,\eta_U]}{\arg \min} \ |J_{\rmf,k}(\eta I_{n}) - 0|, \\
    {V}_{{1,\rm adapt},k} &= \eta_k I_n.
\end{align*}
This establishes that both $\eta_k$ and ${V}_{{2,\rm adapt},k} = 0$ constitute an optimal solution for  \eqref{covmin} when $\SJ_{\rmf,k}$ is empty.
\end{proof}

 


%
%


In Section \ref{sec:rls_var}, we introduced two variations of RLS, resulting in AISE and AISE/VRF-ER, both of which are summarized in Table \ref{Tab:rls_comparison}.
%
\begin{center} 
{\renewcommand{\arraystretch}{2}%
\begin{tabular}{ |l|l|l| }
 \hline
 \textbf{\makecell{AISE \\variation}}  & \textbf{ \makecell{Covariance \\update equation}} \\
 \hline
  AISE  & $P_{k+1}^{-1} = P_k^{-1} + \widetilde{\Phi}_k^\rmT \widetilde{R} \widetilde{\Phi}_k$ \\
 \hline
 AISE/VRF-ER  &   \makecell{$P_{k+1}^{-1} =\lambda_k P_k^{-1}$ \\$ + (1-\lambda_k)R_\infty +  \widetilde{\Phi}_k^\rmT \widetilde{R} \widetilde{\Phi}_k$} \\
 \hline
\end{tabular}
\captionof{table}{\it Variations of AISE with their respective covariance update. $\lambda_k$ is given by \eqref{VRF_eq}.} \label{Tab:rls_comparison}}
\end{center}

\vspace{-5mm}
\section{NUMERICAL EXAMPLES}  \label{sec:num_example}
\vspace{-1mm}
In this section, we compare the performance of AISE with AISE/VRF-ER. We present two illustrative examples: firstly, digital PID control of first-order-lag-plus-dead-time dynamics using adaptive numerical differentiation, and, secondly, velocity estimation using noisy position data from autonomous vehicles for collision avoidance.

\subsection{Digital PID Control} \label{sec:PID_prob_formu}
The continuous-time first-order-lag-plus-dead-time dynamics are given by the transfer function
\begin{align}
	{G}(s) \isdef  \frac{{K}e^{-\tau_d s}}{\tau_c s +1}, 	\label{tf_plant_continuous}
\end{align}
where $K > 0$ is the DC gain, $\tau_c > 0$ is the time constant, and $\tau_d \geq 0$ is the dead time.
The zero-order-hold discretization of \eqref{tf_plant_continuous} is given by 
\begin{align}
	{G}_{\rmd}(z) \isdef  \frac{{K(1-\gamma)}}{z^{n_{\rmd}}(z - \gamma)}. 	\label{tf_plant_discrete_hat}
\end{align}
where $\gamma \isdef e^{\frac{T_\rms}{\tau_c}} \in (0,1)$ and $n_{\rmd} \isdef {\tau_{\rmd}}/{T_{\rms}}.$ 
The PID controller has the discrete-time, linear time-invariant transfer function
\begin{align}
	{G}_{\rmc}(z) \isdef  {K_{\rmP}} + \frac{K_{\rmI}}{z-1} - \frac{K_{\rmD}(z-1)}{T_\rms},
\end{align}
where ${K_{\rmP}}$ is the proportional gain, ${K_{\rmD}}$ is the derivative gain, and ${K_{\rmI}}$ is the integrator gain. The discretized plant ${G}_{\rmd}$ is controlled by the PID controller ${G}_{\rmc}$, whose input is the error $e_k$ and whose output is the control signal $u_k$. The plant output is denoted by $y_k$. In the presence of sensor noise $\eta_k$, the measurement is denoted by the noisy sensor output $y_{\rmm,k} \isdef y_k + \eta_k$.
The servo loop consisting of discrete-time, first-order-lag-plus-dead-time dynamics with a discrete-time PID controller is shown in Figure \ref{fig:PID_block_dig}, where $r_k$ is the command and $e_k \isdef r_k - y_{\rmm,k}$ is the error.

\begin{figure}[h!t]
    \begin{center}
    {\includegraphics[width=0.8\linewidth]{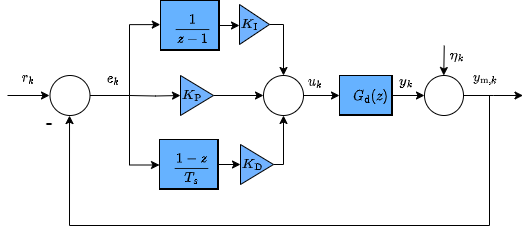}}
    \end{center}\vspace{-3mm}
    \caption{{\it Servo loop consisting of discrete-time first-order-lag-plus-dead-time dynamics with a discrete-time PID controller, where $\eta_{k}$ is the sensor noise.}} 
    \label{fig:PID_block_dig}
\end{figure}

\vspace{-3mm}
The control $u_k$ is written as 
\begin{align}
    u_k \isdef  u_{\rmp,k} + u_{\rmi,k} + u_{\rmd,k}, \label{u_control_input}
\end{align}
where
\begin{align}
    u_{\rmp,k} &\isdef K_{\rmP}e_k, \quad
    u_{\rmi,k} \isdef \frac{K_{\rmI}}{z-1}e_k,\\
    u_{\rmd,k} &\isdef - \frac{K_{\rmD}(z-1)}{T_\rms}e_k = K_{\rmD}\frac{(e_{k}-e_{k-1})}{T_\rms}.  \label{u_derivative}
\end{align}
The derivative action $u_{\rmd,k}$ for the control signal $u_k$ is computed using the backward difference (BD) method applied to the error $e_k$ as shown in \eqref{u_derivative}. In the presence of sensor noise, the error $e_k$ becomes noisy, resulting in a noisy estimate of the derivative of $e_k$ using the BD. We use AISE to estimate the derivative of $e_k$ for computing $u_{\rmd,k}$. 
The servo loop consisting of discrete-time first-order-lag-plus-dead-time dynamics with a discrete-time PID controller combined with adaptive differentiation (PID/AD) 
is shown in Figure \ref{fig:PIAD_block_dig}, where the ``AISE'' indicates that AISE is used to adaptively differentiate $e_k$.
\begin{figure}[h!t]
    \begin{center}
    {\includegraphics[width=0.8\linewidth]{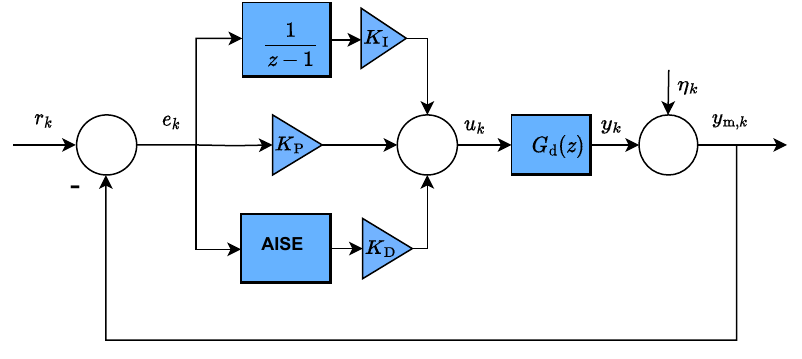}}
    \end{center}\vspace{-3mm}
    \caption{{\it Servo loop consisting of discrete-time first-order-lag-plus-dead-time dynamics with a discrete-time PID controller combined with adaptive differentiation (PID/AD).}}
    \label{fig:PIAD_block_dig}
\end{figure} 
\vspace{-0.0in}

We compare the performance of AISE and AISE/VRF-ER as adaptive differentiators in PID/AD controllers with conventional PID controllers. With conventional PID controllers, we refine the noisy derivative of the error by augmenting BD with Butterworth (BW) and moving-average (MA) filters. To assess the accuracy of the step response, we define the root-mean-square error (RMSE) as
\begin{align}
 {\rm RMSE} \isdef
\sqrt{{\displaystyle\sum_{k=1}^{N}\frac{({y}_{k}-\Bar{y}_{k})^2}{N}}},  \label{rms}
\end{align} 
where $\Bar{y}_k$ denotes the step response of the PID controller using BD in the absence of sensor noise. We treat $\Bar{y}_k$ as ground truth for computing RMSE for various numerical differentiation algorithms.



%
\begin{exam} \label{pid}
      {\it Digital PID control}.
      This example compares the accuracy of the step response of discrete-time first-order-lag-plus-dead-time dynamics under PID control with various numerical differentiation techniques, namely, BD, BD with the MA filter (BD/MA), BD with the BW filter (BD/BW), AISE, and AISE/VRF-ER.
      The key feature of the problem is the presence of nonstationary sensor noise. 

    The parameters in \eqref{tf_plant_continuous} are $K = 1,$ $\tau_{\rmc} = 1,$ $\tau_{\rmd} = 1,$ and $T_{\rms} = 0.01$ sec.
    The gains of the PID controller are $K_{\rmp} = 1.5,$ $ K_{\rmi} = 1.0,$ $ K_{\rmd} = 0.25$.
    The nonstationary sensor noise is given by $D_2v$, where $v$ is white Gaussian noise and $D_2 = 1.5$ for $1 \leq k \leq 1999$ and $D_2 = 1$ for $2000 \leq k \leq 3501$, which results in a signal-to-noise ratio (SNR) of 8.87 dB and 12.22 dB, respectively.
    %
    %
    %
    %
    %



    For the MA filter, the window size is 10 data points, and the BW filter is $5^{th}$ order with a cutoff frequency of $0.6\pi$ rad/step. 
    For AISE, let $n_\rme = 12$, $n_\rmf = 20,$ $ R_z = 1, R_d = 10^{-7}, R_\theta = 10^{-0.1}I_{25}$, and ${V_{1},V_{2}}$ are adapted, where $\eta_{\rmL} = 10^{-6}$, $\eta_{\rmU} = 10^{-2}$, and $\beta = 0.55$ as in Section \ref{sec:AdapInptStateEst}.
    For AISE/VRF-ER,  the parameters are the same as those of AISE, and, for VRF-ER, $\eta = 0.5, t_n = 20, t_d = 80, \alpha = 0.08$, and $R_{\infty} = 50.$ 

    In Figure \ref{fig:PI_res}, the step response $y_k$ of the plant and the control input $u_k$ with the PI controller are presented, demonstrating their behavior in the absence of sensor noise. 
    Figure \ref{fig:pid_res_absence} shows the step response $\Bar{y}_k$ of the plant and control input $u_k$ with the PID controller in the absence of the sensor noise. 

    Figure \ref{fig:pid_res_pres_BD} shows the step response with the PID controller in the presence of sensor noise using BD for numerical differentiation. The step response is noisy because of the numerical differentiation of the error signal $e_k$ using BD. 
    Figure \ref{fig:pid_res_pres_MA} shows the step response with the PID controller in the presence of sensor noise using BD/MA for numerical differentiation. The step response is less noisy than the step response of the PID controller with BD differentiation because $e_k$ is numerically differentiated using BD and filtered with the MA which averaged the noise. 
    Figure \ref{fig:pid_res_pres_BW} shows the step response with the PID controller in the presence of sensor noise using BD/BW for numerical differentiation.

    Figure \ref{fig:pid_res_pres_aise} shows the step response with the PID/AD controller in the presence of sensor noise. The numerical differentiation of the error signal $e_k$ is computed using AISE. 
    Figure \ref{fig:theta_pid_aise} shows the eigenvalues of the covariance matrix $P_k$ and the parameter vector $\theta_k$ of AISE. The eigenvalues of $P_k$ decrease monotonically and the $\theta_k$ have not converged.   
    Figure \ref{fig:pid_res_pres_aise_vrf_er} shows the step response with the PID/AD controller in the presence of sensor noise. The numerical differentiation of the error signal $e_k$ is computed using AISE/VRF-ER.
    Figure \ref{fig:theta_pid_aise_vrf_er} shows the eigenvalues of the covariance matrix $P_k$ and the parameter vector $\theta_k$ of AISE/VRF-ER. ER guarantees that, for all $k \ge 0$, $\eigmax(P_{k+1}) \le \max \{ \eigmax(P_k),\eigmax(R_\infty^{-1}) \}$, where $\eigmax(R_\infty^{-1})$ is shown by the red dashed line in Figure \ref{fig:theta_pid_aise_vrf_er}(a) (see Corollary 1 of \cite{lai2022exponential}).    
    Hence, the eigenvalues of $P_k$ are eventually bounded by $\eigmax(R_\infty^{-1}),$ and the parameters $\theta_k$ converge.
    Figure \ref{fig:theta_pid_aise_vrf_er}(b) shows the variable-rate-forgetting factor $\lambda_k$ for the AISE/VRF-ER. The  RMSE \eqref{rms} for the step responses obtained using various numerical differentiation methods are summarized in Table \ref{Tab:step_response_RMSE}.

\vspace{4mm}
\begin{center} 
{\renewcommand{\arraystretch}{2}%
\begin{tabular}{ |l|l| }
 \hline
 \textbf{\makecell{Method}}  & \textbf{ \makecell{RMSE}} \\
 \hline
 BD &  $0.1904$\\
 \hline
  BD/MA &   $0.1302$ \\
 \hline
  BD/BW &  $0.2830$\\
 \hline
  AISE &  $0.1201$ \\
   \hline
  AISE/VRF-ER &  $0.0998$ \\
 \hline
\end{tabular}
\captionof{table}{\it RMSE of the step response of the PID controllers for several numerical differentiation methods.} \label{Tab:step_response_RMSE}}
\end{center}

          \begin{figure}[h!t]
              \begin{center}{\includegraphics[width=1\linewidth]{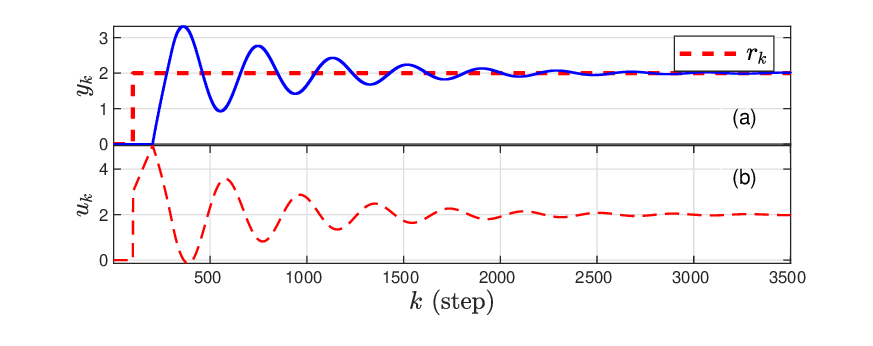}}
            \end{center}\vspace{-3mm}
            \caption{{\it  Example \ref{pid}: (a) shows the step response $y_k$ of the PI controller in the absence of sensor noise, where, without the derivative action, the settling time is 20 sec. (b) shows the control $u_k$.}} 
            \label{fig:PI_res}
          \end{figure} 
          \vspace{-3mm}
          \begin{figure}[h!t]
              \begin{center}{\includegraphics[width=1\linewidth]{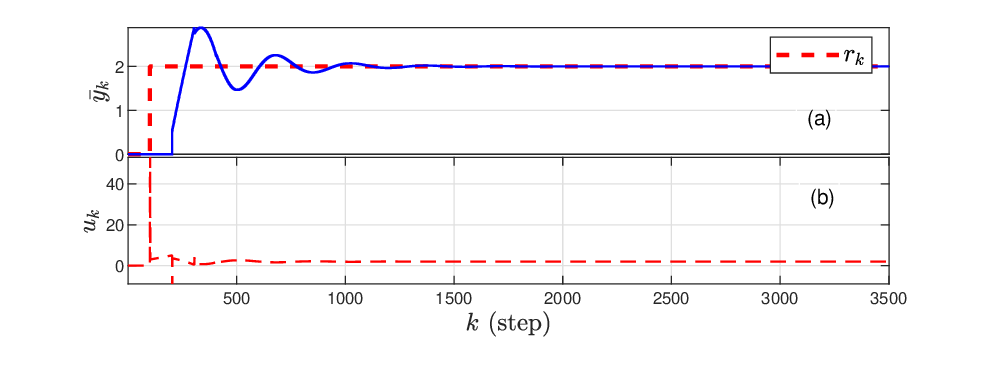}}
            \end{center}\vspace{-3mm}
            \caption{{\it  Example \ref{pid}: (a) shows the step response $\Bar{y}_k$ of the PID controller in the absence of sensor noise using BD for numerical differentiation. (b) shows the control $u_k$.}} 
            \label{fig:pid_res_absence}
          \end{figure}
          \begin{figure}[h!t]
            \begin{center}{\includegraphics[width=1\linewidth]{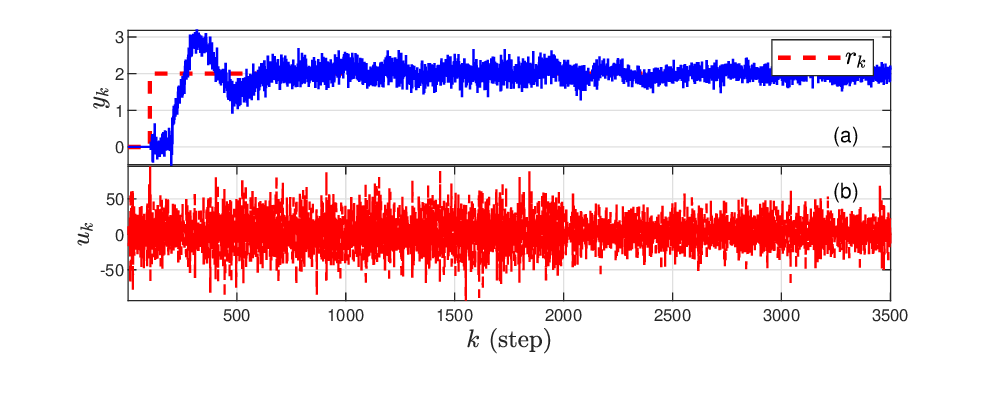}}
          \end{center}\vspace{-3mm}
          \caption{{\it  Example \ref{pid}: (a) shows the step response $y_{k}$ of the PID controller in the presence of sensor noise using BD for numerical differentiation. The steady-state response is noisy. The ${\rm RMSE}$ is $0.1904$. (b) shows the control $u_k$.}} 
          \label{fig:pid_res_pres_BD}
          \end{figure}
                    \begin{figure}[h!t]
            \begin{center}{\includegraphics[width=1\linewidth]{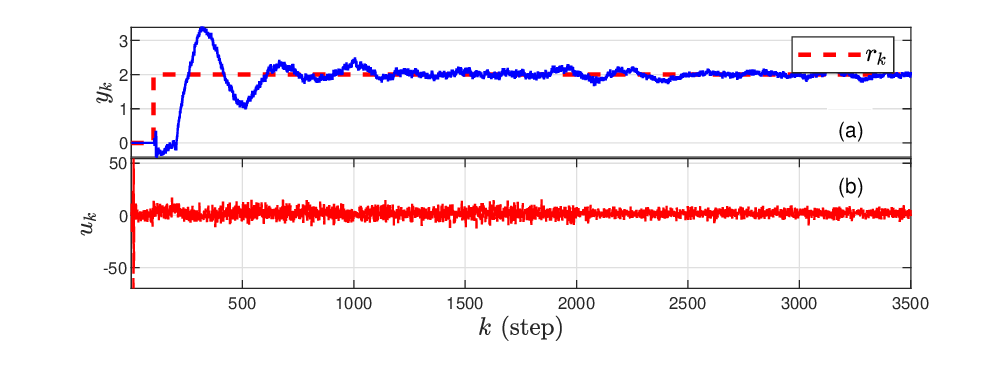}}
          \end{center}\vspace{-3mm}
          \caption{{\it  Example \ref{pid}: (a) shows the step response $y_{k}$ of the PID controller using BD/MA for numerical differentiation, where the MA filter of window size of 10 data points in the presence of sensor noise. The ${\rm RMSE}$ is $0.1302$. (b) shows the control $u_k$.}} 
          \label{fig:pid_res_pres_MA}
          \end{figure}
                    \begin{figure}[h!t]
            \begin{center}{\includegraphics[width=1\linewidth]{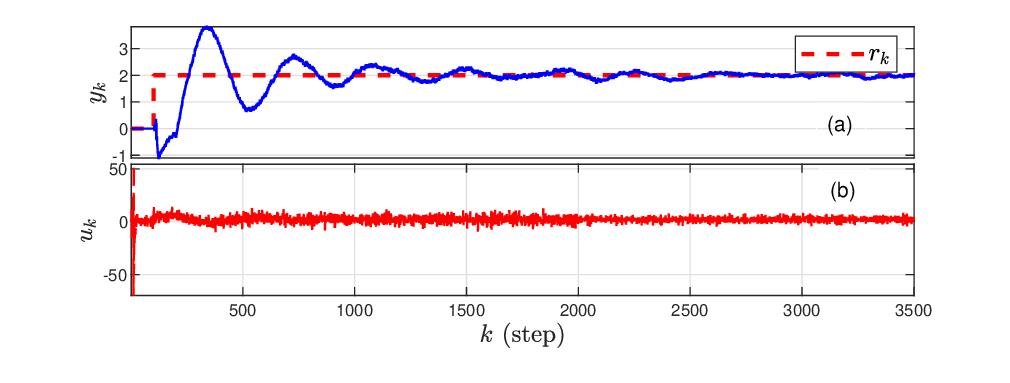}}
          \end{center}\vspace{-3mm}
          \caption{{\it  Example \ref{pid}: (a) shows the step response $y_{k}$ of the PID controller using BD/BW for numerical differentiation, where the BW filter is of  $5^{\rmt \rmh}$ order with a cutoff frequency of $0.6\pi$ rad/step. in the presence of sensor noise. The ${\rm RMSE}$ is $0.2830$. (b) shows the control $u_k$.}} 
          \label{fig:pid_res_pres_BW}
          \end{figure}
           \begin{figure}[h!t]
            \begin{center}{\includegraphics[width=1\linewidth]{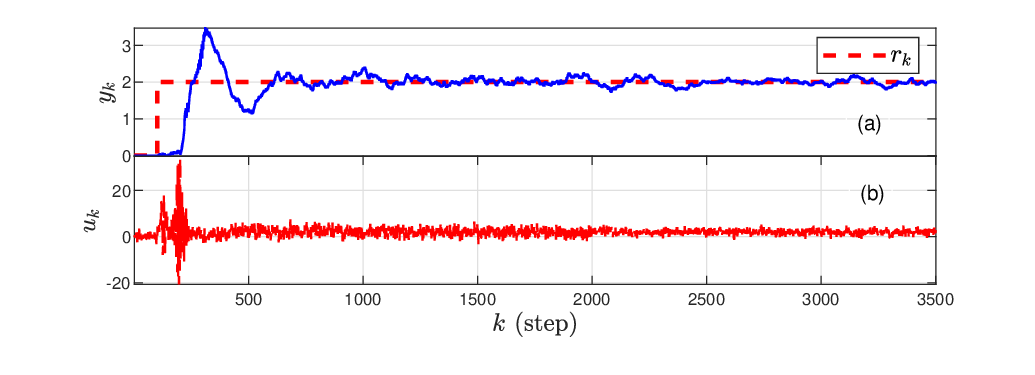}}
          \end{center}\vspace{-3mm}
          \caption{{\it  Example \ref{pid}: (a) shows the step response $y_{k}$ of the PID/AD controller with AISE in the presence of sensor noise. The ${\rm RMSE}$ is $0.1201$. (b) shows the control $u_k$.}} 
          \label{fig:pid_res_pres_aise}
          \end{figure}
          \begin{figure}[h!t]
            \begin{center}{\includegraphics[width=1\linewidth]{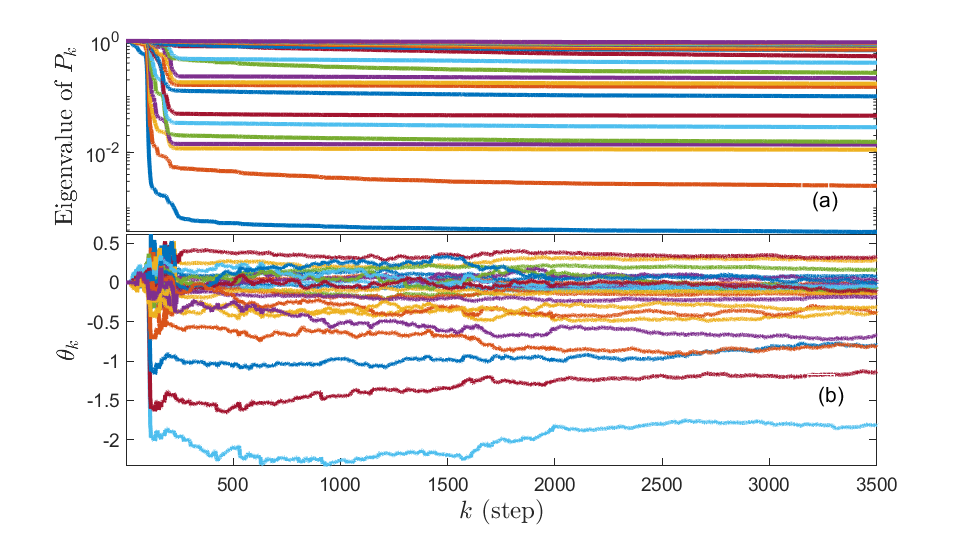}}
          \end{center}\vspace{-3mm}
          \caption{{\it  Example \ref{pid}: (a) shows the eigenvalues of $P_k$ of AISE. (b) shows the estimated coefficient vector $\theta_k$ of AISE.}} 
          \label{fig:theta_pid_aise}
          \end{figure}
           \begin{figure}[h!t]
            \begin{center}{\includegraphics[width=1\linewidth]{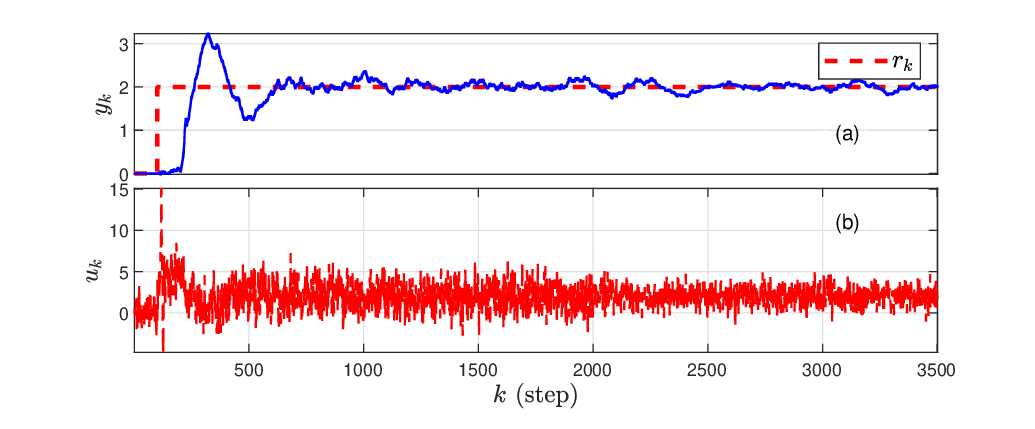}}
          \end{center}\vspace{-3mm}
          \caption{{\it  Example \ref{pid}:(a) shows the step response $y_{k}$ of the PID/AD controller with AISE/VRF-ER in the presence of sensor noise. The ${\rm RMSE}$ is $0.0998$. (b) shows the control $u_k$.}} 
          \label{fig:pid_res_pres_aise_vrf_er}
          \end{figure}
          \begin{figure}[h!t]
            \begin{center}{\includegraphics[width=1.05\linewidth]{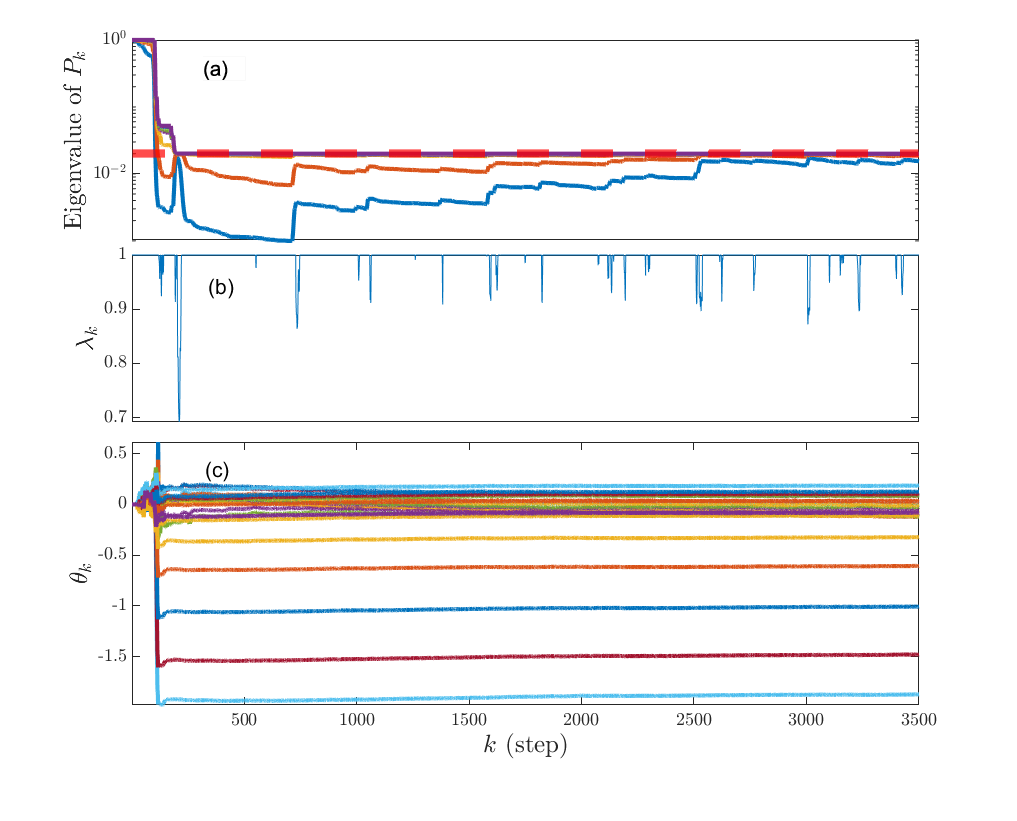}}
          \end{center}\vspace{-5mm}
          \caption{{\it  Example \ref{pid}: (a) shows the eigenvalues of $P_k$ of AISE/VRF-ER. The dashed red line shows $\eigmax(R_\infty^{-1})$. (b) shows the forgetting factor $\lambda_k$ in AISE/VRF-ER. (c) shows the estimated coefficient vector $\theta_k$ of AISE/VRF-ER.}} 
          \label{fig:theta_pid_aise_vrf_er}
          \end{figure}
          %
          %
          
\end{exam}

\subsection{Target tracking for collision avoidance} \label{subsec: carsim example}

To estimate the relative velocity of the target vehicle, AISE, and AISE/VRF-ER are applied to simulated position data of a target vehicle relative to a host vehicle. The CarSim simulator is used to simulate a scenario (depicted in Figure \ref{fig:Carsim_Pic}) in which an oncoming target vehicle (white van) slides over to the wrong lane. The host vehicle (blue van) performs an evasive maneuver to avoid a collision. Differentiation of the relative position data along the global $y$-axis (shown in Figure \ref{fig:Carsim_Pic}) is done to estimate the relative velocity along the same axis. The same method yields an estimate of the relative velocity along the global $x$-axis (not shown). We compare the performance of AISE and AISE/VRF-ER. The estimated relative velocity of the target vehicle can be used for collision avoidance purposes. 
          \begin{figure}[h!t]
              \begin{center}{\includegraphics[width=0.4\linewidth]{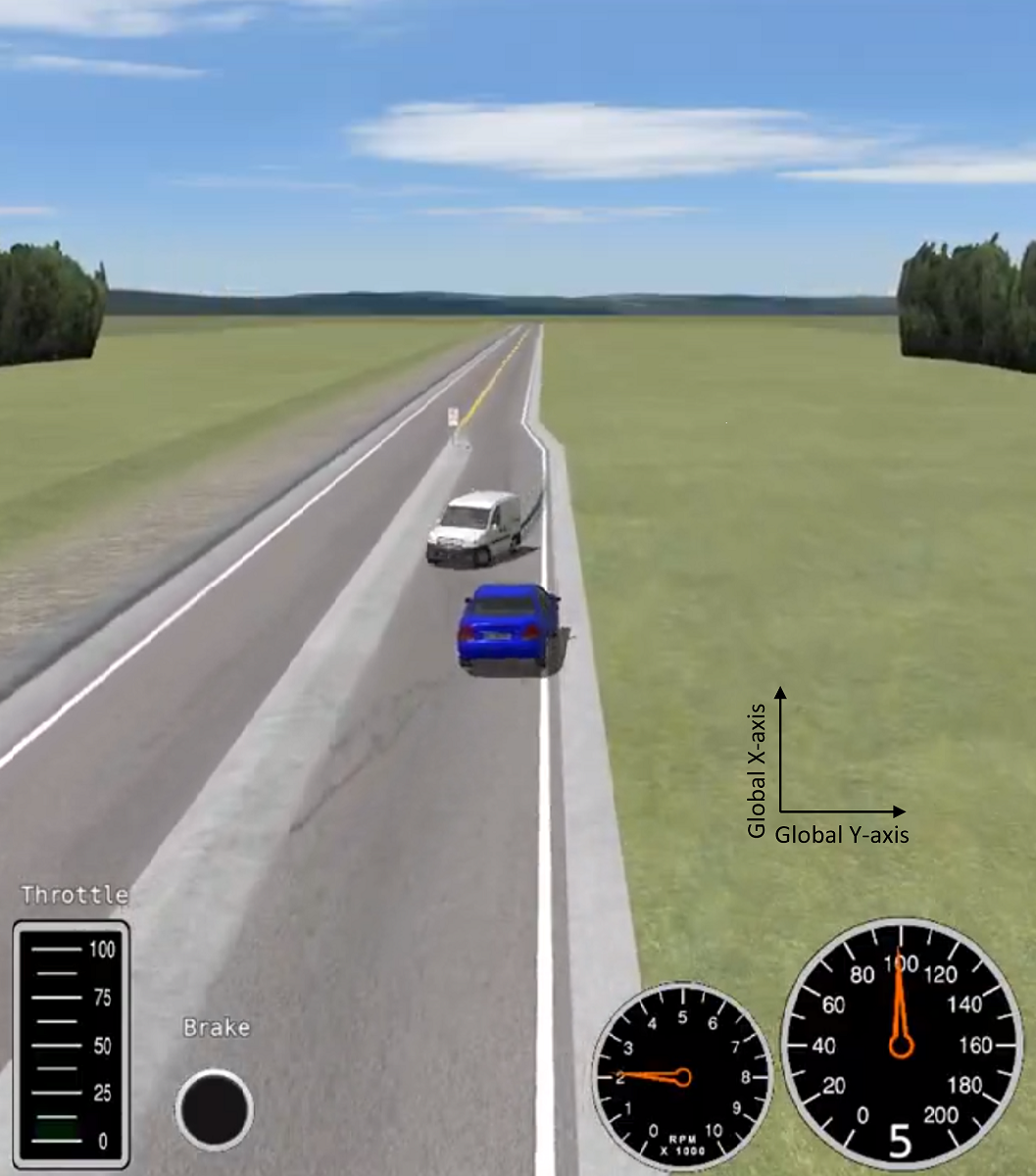}}
            \end{center}\vspace{-3mm}
            \caption{{\it Collision-avoidance scenario in CarSim}. In this scenario, the oncoming vehicle (the white van) enters the opposite lane, and the host vehicle (the blue van) performs an evasive maneuver to avoid a collision.} 
            \label{fig:Carsim_Pic}\vspace{-0.2mm}
          \end{figure} 
          \vspace{-0.2mm}
\begin{exam} \label{eg_carsim}
      {\it Differentiation of the CarSim data in the presence of stationary sensor noise.}
      The relative position data is corrupted with white Gaussian stationary sensor noise with SNR 40 dB. 
    For AISE, let $n_\rme = 25$, $n_\rmf = 50,$ $ R_z = 1, R_d = 10^{-6.7}, R_\theta = 10^{-0.1}I_{25}$, ${V_{1},V_{2}}$ are adapted, where $\eta_{\rmL} = 10^{-6}$, $\eta_{\rmU} = 10^{-2}$, and $\beta = 0.55$.  For AISE/VRF-ER  the parameters are the same as those of AISE and for VRF-ER $\eta = 0.8, t_n = 20, t_d = 80, \alpha = 0.08$, and $R_{\inf} = 10.$

    Figure \ref{fig:estimate_aise_carsim} compares the true first derivative with the estimates obtained from the AISE and AISE/VRF-ER. The estimate generated using AISE/VRF-ER is more accurate than the estimate generated using AISE. 
    Figure \ref{fig:theta_AISE_carsim} shows the eigenvalues of the covariance matrix $P_k$ and the parameter vector $\theta_k$ of AISE. The eigenvalues of $P_k$ decrease monotonically due to which AISE is not adapting to changes in the signal. 
    In Figure \ref{fig:theta_AISE_VRF-ER_carsim},   the eigenvalues of the covariance matrix $P_k$, the variable-rate-forgetting factor $\lambda_k$, and the parameter vector $\theta_k$ for AISE/VRF-ER are displayed. ER ensures that, for all $k \geq 0,$ $\eigmax(P_{k+1}) \le \max \{ \eigmax(P_k),\eigmax(R_\infty^{-1}) \}$, where $\eigmax(R_\infty^{-1})$  is indicated by the red dashed line in Figure \ref{fig:theta_AISE_VRF-ER_carsim}(a) (see Corollary 1 of \cite{lai2022exponential}).
    
    \vspace{-0.5mm}

          \begin{figure}[h!t]
              \begin{center}{\includegraphics[width=1\linewidth]{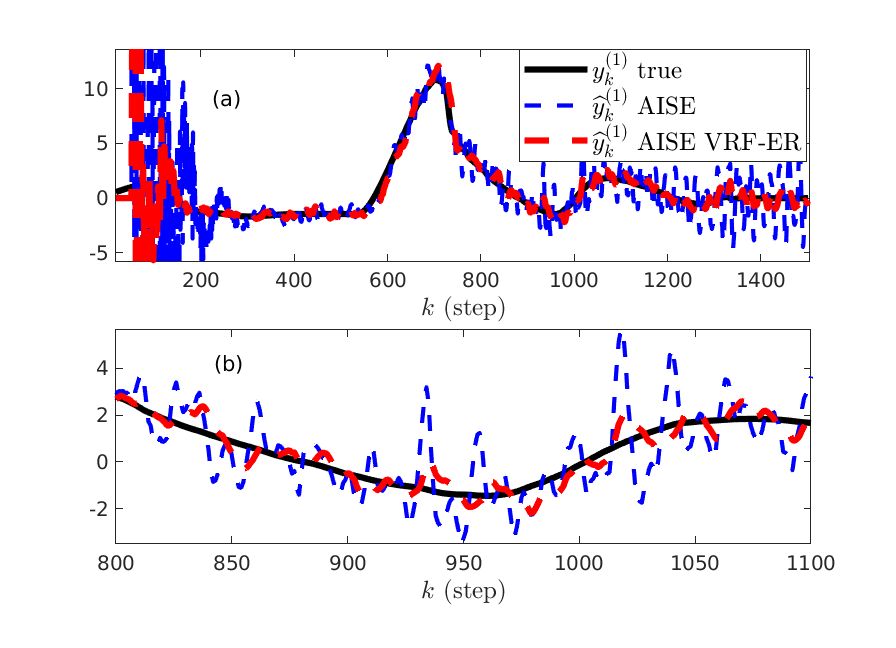}}
            \end{center}\vspace{-3mm}
            \caption{\textit{Example \ref{eg_carsim}: (a) The numerical derivatives estimated by AISE and AISE/VRF-ER follow the true first derivative $y^{(1)}_k$ after an initial transient of 200 steps.
        (b) shows a zoom of (a). At steady state, AISE/VRF-ER is more accurate than AISE. The SNR is 40 dB.}} 
            \label{fig:estimate_aise_carsim}\vspace{-0.2mm}
          \end{figure} 
          \vspace{-0.2mm}
                    \begin{figure}[h!t]
              \begin{center}{\includegraphics[width=1\linewidth]{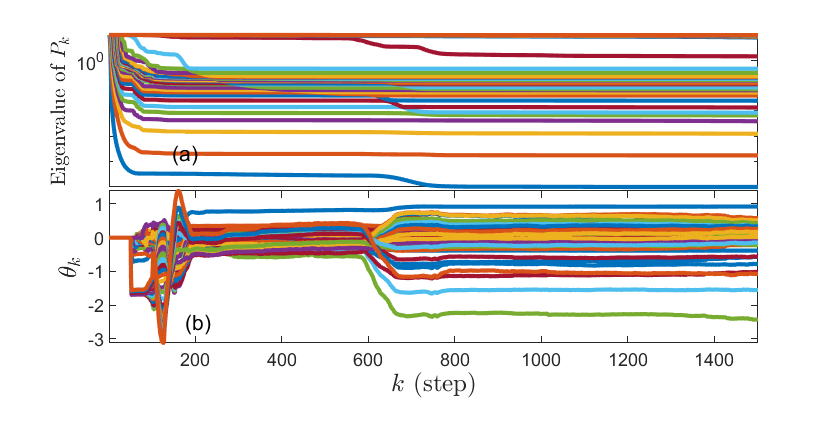}}
            \end{center}\vspace{-3mm}
            \caption{\textit{Example \ref{eg_carsim}: (a) shows the eigenvalues of $P_k$ of AISE. (b) shows the estimated coefficient vector $\theta_k$ of AISE.}} 
            \label{fig:theta_AISE_carsim}
          \end{figure} 
                    \begin{figure}[h!t]
              \begin{center}{\includegraphics[width=1.05\linewidth]{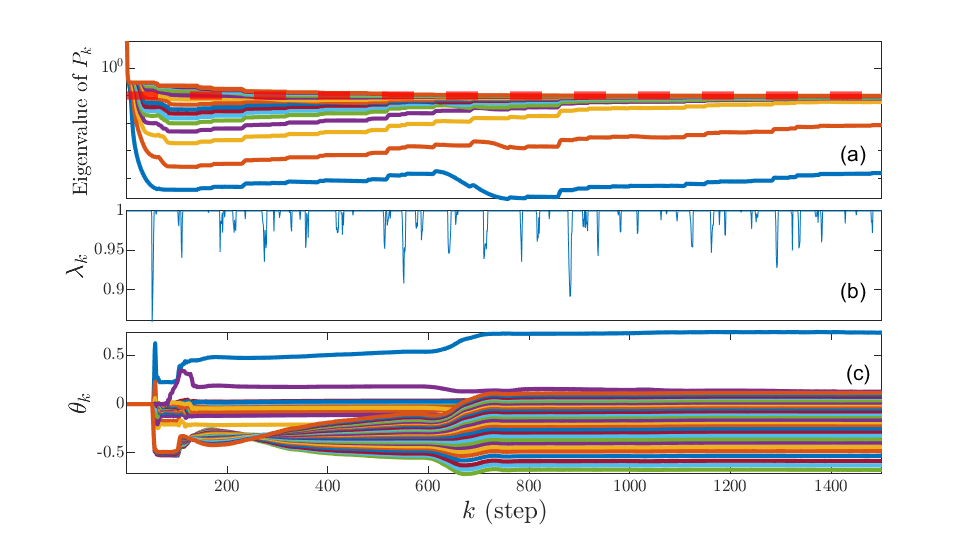}}
            \end{center}\vspace{-3mm}
            \caption{\textit{Example \ref{eg_carsim}: (a) shows the eigenvalues of $P_k$ of AISE/VRF-ER. The dashed red line shows $\eigmax(R_\infty^{-1})$. (b) shows the forgetting factor $\lambda_k$ of AISE/VRF-ER. (b) shows the estimated coefficient vector $\theta_k$ of AISE/VRF-ER. }} 
            \label{fig:theta_AISE_VRF-ER_carsim} \vspace{-0.10in}
          \end{figure} 
          %
          %
\end{exam}

\section{CONCLUSIONS}
This paper extended adaptive input and state estimation (AISE) by incorporating recursive least squares (RLS) with variable-rate forgetting and exponential resetting (VRF-ER). With VRF-ER, RLS uses the $F$-test for variable-rate forgetting as well as exponential resetting to constrain the eigenvalues of the error covariance matrix. 
%
%
AISE and AISE/VRF-ER were both used in digital PID control. We assessed the performance of these methods by considering the step response under the influence of sensor noise.
The performance of AISE and AISE/VRF-ER was compared with PID controllers that incorporate moving average and Butterworth filters to mitigate the noisy derivative control action.
Additionally, we demonstrated the ability of AISE and AISE/VRF-ER to estimate the relative velocity of a target vehicle based on noisy relative position data. These estimates are  potentially useful for enhancing collision-avoidance systems in autonomous vehicles.
Numerical examples showed that AISE/VRF-ER provides improved performance compared to AISE.
%


\section*{ACKNOWLEDGMENTS}
This research was supported by NSF grant CMMI 2031333.

\bibliography{bibpaper}
\bibliographystyle{ieeetr}

\end{document}